\documentclass[conference,letterpaper]{IEEEtran}
\usepackage[utf8]{inputenc} 
\usepackage[T1]{fontenc}

\usepackage{ifpdf}

\usepackage{cite}

\ifCLASSINFOpdf
  \usepackage[pdftex]{graphicx}
\else
\fi
\usepackage[cmex10]{amsmath}
\usepackage{amsfonts}
\usepackage{amssymb}
\usepackage{amsthm}
\usepackage{color}
\usepackage{algorithmic}
\usepackage{array}
\usepackage{url}
\newtheorem{theorem}{Theorem}

\newtheorem{example}{Example}

\newtheorem{lemma}{Lemma}
\newtheorem{remark}{Remark}

\newcommand{\bx}{\mathbf{x}}

\newcommand{\cC}{\mathcal{C}}
\newcommand{\cH}{\mathcal{H}}

\newcommand{\cP}{\mathcal{P}}
\newcommand{\cQ}{\mathcal{Q}}

\newcommand{\cB}{\mathcal{B}}
\newcommand{\cR}{\mathcal{R}}
\newcommand{\cX}{\mathcal{X}}
\newcommand{\cZ}{\mathcal{Z}}

\newcommand{\bX}{\mathbf{X}}

\newcommand{\continue}{\mathit{continue}}

\DeclareMathOperator{\erf}{erf}

\begin{document}
%
\title{Classification in a Large Network}
\author{\IEEEauthorblockN{Vinay A. Vaishampayan}
\IEEEauthorblockA{City University of New York-College of Staten Island\\Staten Island, NY 10314, USA}
}


%

\maketitle

\begin{abstract}
We construct  and analyze the communication cost of protocols (interactive and one-way) for classifying $\bX=(X_1,X_2,\ldots,X_n) \in [0,1)^n \subset \mathbb{R}^n$,  in a  network with $n\geq 2$ nodes, with $X_i$ known only at node $i$. The classifier takes the form $\sum_{i=1}^nh_iX_i \gtrless a$, with weights $h_i \in \{-1,+1\}$.  The interactive protocol (a zero-error protocol) exchanges a variable number of messages depending on the input $\bX$ and its sum rate is directly proportional to its mean stopping time. An exact analysis, as well as an approximation of the mean stopping time is presented and  shows that it  depends on $\gamma=\alpha+(1/2-\beta)$, where $\alpha=a/n$ and $\beta=m/n$, with $m$ being the number of positive weights. In particular, the mean stopping time grows logarithmically in $n$ when $\gamma=0$, and is bounded in $n$ otherwise.
Comparisons show that the sum rate of the  interactive protocol is smaller than that of the one-way protocol when the error probability for the one-way protocol is small, with the reverse being true when the error probability is large. Comparisons of the interactive protocol are also made with lower bounds on the sum rate.
\end{abstract}



%
\IEEEpeerreviewmaketitle
\section{Introduction}

We study the communication cost of implementing a 2-classifier $f~:~\mathbb{R}^n\rightarrow \{0,1\}$ which maps input vector $\bX=(X_1,X_2,\ldots,X_n)$ to a class label in the set $\{0,1\}$. It is assumed that the classifier is to be implemented in a network with $n$ sensor nodes, where $X_i$ is a random variable known at the $i$th node alone.  We propose a distributed algorithm for solving the classification problem and  provide an analysis of the communication rate required by this algorithm. We also investigate lower bounds on the communication rate. 

Previous related work is summarized in Sec.~\ref{sec:previous}, problem setup is  in Sec.~\ref{sec:setup}, communication protocols and their analysis are  in Secs.~\ref{sec:bitexchange}--\ref{sec:SingleRound}, lower bounds on the communication rate are in Sec.~\ref{sec:converse}, numerical performance results are  in Sec.~\ref{sec:numeric} and Sec.~\ref{sec:summary} contains a summary of the paper.

\section{Background}
\label{sec:previous}
Given a function  $f~:~\cX_1\times\cX_2\ldots\times \cX_n \rightarrow \cZ$, where $\cX_i$ and $\cZ$ are arbitrary sets, its   \emph{communication complexity}~\cite{Yao:1979},  is $\cC(f):=\inf_{\cP} R(\cP)$,  where $R(\cP)$ is the communication cost of protocol $\cP$ for computing $f$ in a distributed environment.  Here, the \emph{average-case} rather than worst-case $R(\cP)$ is considered, the average being over all inputs with respect to a known probability distribution~\cite{OrEl:1990}.

Communication complexity for interactive function computation of Boolean functions is explored in great detail in the seminal work~\cite{OrEl:1990}, which has also influenced the development here. Our paper addresses a problem of current interest since classification is an important part of machine learning. Interactive and one-way hypothesis testing were studied information theoretically (i.e. as the block length grows) in \cite{XiangKim:2012} for $n=2$ nodes and $r\leq 2$ rounds. Here, we assume a block length of unity, the methods used are from sequential analysis, and the results derived are for infinite round protocols for $n\geq 2$ nodes. We note that classification is also a central operation in nearest lattice point search~\cite{BVC:2017},~\cite{VB:2017}. The bit exchange protocol described here is in a class of protocols  described in~\cite{Orlitsky:1992}. 


\section{Problem Setup}
\label{sec:setup}
Since $X_i$ is known at node $i$ and $h_i$ is known globally, the quantity $|h_i|X_i$ can be treated as a new $X_i$, with a modified probability distribution, known to all nodes. If $h_i=0$ for some $i$, it can be dropped and treated as a problem with a smaller $n$. Thus it suffices to consider classifiers with filter coefficients $h_i\in \{-1,+1\}$, with magnitudes of the $h_i$'s absorbed into the probability distribution $p_\bX$.  Let $\cH^+:=\{i:h_i=+1\}$ and $\cH^-:=\{i~:~h_i =-1\}$. Let $m:=|\cH^+|$, the number of  positive filter weights. Then $|\cH^-|=n-m$. Let
\begin{equation}
X_s:=\sum_{i\in \cH^+}X_i-\sum_{i\in \cH^-}X_i.
\label{eqn:Xs}
\end{equation}
Our classification function 
\begin{equation}
f(\bX):=\left\{ \begin{array}{cc}
1, & X_s \geq  a,  \\
0, &   X_s <  a,
\end{array} \right.
\label{eqn:problem}
\end{equation}
partitions $\mathbb{R}^n$ into two classes $f^{-1}(0)=\{\bX~:~X_s <a \}$  and its complement $f^{-1}(1)$.

We assume here that $X_i$ are independent, identically distributed (iid) random variables, uniformly distributed on the unit interval $[0,1)$.   Note that under our assumptions for the source distribution, the classification problem is  non-trivial only if $a$ lies in the two-side open interval $(-(n-m),m)$.

\section{The Infinite Round Bit-Exchange Protocol}
\label{sec:bitexchange}
Nodes exchange information in a pre-arranged manner and any information transmitted by node-$i$ at time $t$ can depend only on $X_i$ and  information that it has received from other nodes at times $t' < t$. We do not assume a broadcast model for counting bits, so every transmitted bit  is indexed by $i,j$ and $t$, the source, destination and transmission time, respectively. The  cost  is the sum of communication over $i,j,t$. 
 
The   Interactive Protocol  computes $f(\bX)$  as follows. A node is selected as a leader node. Here, we assume that the leader node is distinct from the $n$ sensor nodes~\footnote{As will be apparent later, we could have chosen one of the sensor nodes as a leader node at a slightly lower communication cost. The corrections required for this are simple and negligible  $n\to \infty$ relative to the communication cost. } and at the outset of a session has no access to $X_i,~i=1,2,\ldots,n$. The order of communication is known to all nodes.  The time axis is broken into sessions, a session starts with a fresh observation $\bX$ and concludes when $f(\bX)$ is computed. Here we are concerned with  a single session.  A session is divided up into  rounds.

Node-$i$ produces a bit stream $b(i,j),~j=1,2,\ldots$  according to the following standard binary expansion rule: set $X(i,1)=X_i$ and for $j=1,2,\ldots...$ compute
\begin{eqnarray}
b(i,j) & = & \lfloor 2 X(i,j)\rfloor, \nonumber \\
X(i,j+1) & =& 2X(i,j)-b(i,j).
\label{eqn:bitdynam}
\end{eqnarray}
Let
\begin{equation}
B(l):=\sum_{i\in \cH^+} b(i,l)-\sum_{i\in \cH^-}b(i,l).
\end{equation}
and
$$Z(j):=\sum_{l=1}^j B(l)2^{-l},$$ with $Z(0)=0$.
In round $j\geq 1$, node-$i$ sends $b(i,j)$ to the leader node and after the leader node receives all the bits $b(i,j)$, $i=1,2,\ldots,n$, it computes $Z(j)$
and sends $s(j)$ back to each sensor node, where
\begin{equation}
s(j)=\left\{\begin{array}{cc}
0, & Z(j) \leq  L(j)  \\
1, & Z(j) \geq  U(j) \\
\continue, & \mbox{otherwise},
\end{array} \right.
\label{eqn:continue}
\end{equation}
with $L(j):=a-m2^{-j}$ and $U(j):=a+(n-m)2^{-j}$. Observe that $U(j+1)=U(j)-(n-m)2^{-(j+1)}$ and $L(j+1)=L(j)+m2^{-(j+1)}$, $j\geq 1$.

Let $T$ be the stopping time, i.e. $T$ is the smallest $k$ for which $s(k)\neq \continue$.
Since the leader node sends back  $n\log_2(3)$ bits at the conclusion of every round, 
the  total number of bits communicated  in a session is 
$
R_{sum} =  n(1+log_2(3))\bar{T}
$
where $\bar{T}:=E[T]$ is the mean stopping time, over all inputs $\bX$.  
\begin{lemma}
If at the end of round $j$, $s(j)\neq \continue$, then either $f(\bX)=0$ or $f(\bX)=1$.
\end{lemma}

\begin{proof}
This follows from the fact that $Z(j)-(n-m)2^{-j} < X_s < Z(j)+m2^{-j}$ and the definition of $f$ in \eqref{eqn:problem}.
\end{proof}

\subsection{A Recursive Description of the Protocol}
For the purpose of analysis, it is convenient to have an alternate description, referred to as the \emph{recursive} description,  of the protocol.  Towards this end consider the following alternate description of the feedback signal from leader to sensor.
Let 
\begin{equation}
\tilde{s}(j)=\left\{\begin{array}{cc}
0, & B(j) \leq \tilde{L}(j), \\
1, & B(j) \geq \tilde{U}(j), \\
\continue, & \mbox{otherwise}.
\end{array} \right.
\label{eqn:continue1}
\end{equation}
where $\tilde{L}(j)=2a(j)-m$, $\tilde{U}(j)=2a(j)+(n-m)$, 
\begin{equation}
a(j+1)=2a(j)-B(j),~~j\geq 1 \mbox{ and } a(1)=a.
\label{eqn:threshold}
\end{equation}
Observe that $\tilde{L}(j+1)=2\tilde{L}(j)+m-2B(j)$ and 
$\tilde{U}(j+1)=2\tilde{U}(j)-2B(j)-(n-m)$, $j\geq 1$.
\begin{remark}
From the form of $\tilde{L}(j)$ and $\tilde{U}(j)$, it is clear that at iteration $j$, the effective threshold is $a(j)$. Thus for the recursive protocol, the threshold shifts with each iteration, but the width of the interval $\tilde{U}(j)-\tilde{L}(j)$ remains fixed, as opposed to the non-recursive description, where $U(j)-L(j)$ decreases with $j$.
\end{remark}
\begin{lemma}$\tilde{s}(j)=s(j),~j \geq 1$ \end{lemma}
\begin{proof}
Observe that $s(j)$ is determined by testing $Z(j)$ against thresholds $U(j),L(j)$. Similarly $\tilde{s}(j)$ is obtained by testing $B(j)$ against $\tilde{U}(j),\tilde{L}(j)$. Since $Z(j)=Z(j-1)+2^{-j}B(j)$, it suffices to show  that 
\begin{equation}
U(j)=2^{-j}\tilde{U}(j)+Z(j-1),~~j\geq 1
\label{eqn:suffices}
\end{equation}
with a similar statement for $L(j)$ and $\tilde{L}(j)$.
Eq.~\eqref{eqn:suffices} is clearly true for $j=1$ since $U(1)=a+(n-m)/2$ and $\tilde{U}(1)=2a+(n-m)$ and $Z(0)=0$. Assume it holds for $j$. Then
 \begin{eqnarray}
 U(j+1) & = & U(j)-(n-m)2^{-(j+1)}\nonumber \\
 & = & 2^{-j}\tilde{U}(j)+Z(j-1)-(n-m)2^{-(j+1)} \nonumber \\
 & = & 2^{-(j+1)}\tilde{U}(j+1)+Z(j-1)+B(j)2^{-j} \nonumber \\
 & = & 2^{-(j+1)}\tilde{U}(j+1)+Z(j).
 \end{eqnarray}
\end{proof}

For the analysis presented later we observe that 
\begin{lemma}
\begin{equation}
Pr[B(l)=k]=\left\{ \begin{array}{cc} {n \choose k+n-m}, & k=-(n-m),\ldots,m, \\
0, & \mbox{otherwise}.
\end{array} \right.
\label{eqn:weight}
\end{equation}
\end{lemma}
\begin{proof}
To see that \eqref{eqn:weight} is true, let  $B(l)=w^+-w^-$ where $w^+$ is the number of 1's in positions $\cH^+$ and $w^-$ the number of $1$'s in the positions $\cH^-$. Then the number of $0$'s in the positions $\cH^-$ is  $(n-m)-w^-$ and  $B(l)=k$ if and only if $w^+ + (n-m)-w^-=k+(n-m)$.
\end{proof}

In order to understand how $R_{sum}$ scales with $n$ we now proceed to analyze $\bar{T}$.  We first present a computational approach  when $a$ is an integer, followed by an asymptotic analysis for general $a$.

\subsection{Exact  Analysis}
We first present the  analysis  for $n=2,~m=1$ where the only non-trivial case  is $a=0$.  Node 1 sends a bit to node 2, which responds with a signal to continue or stop (here and for this case only, the leader node is a sensor node; see footnote on previous page). When the signal is to stop, both nodes know the class. Thus the average rate is $2\bar{T}=4\mbox{ bits}$, which follows from the fact that $Pr[T>k]=2^{-k}$ and $\bar{T}=1+\sum_{k>0}Pr[T>k]$.  

A general analysis for   $n\geq 2$ is presented next. We work with the recursive description of the protocol. The basic idea is that if the initial threshold $a$ is an integer then so is $a(j),j\geq 1$ as can be seen from \eqref{eqn:threshold}. Together with the fact that the distribution of $\{X(i,j),~i=1,2,\ldots,n\}$ does not depend on iteration index $j$ means that  $a(j)$ can be regarded as a state variable and if it shown to be restricted to a finite set, then $\bar{T}(a)$, the mean stopping time associated with a threshold $a$, can be obtained through the state transition probability matrix of a finite state transition graph. Let $$\bar{\mathbf T}=(\bar{T}(-(n-m)+1),\ldots,\bar{T}(m-1))^t$$

\begin{theorem}
The vector of mean stopping times is a solution of
\begin{equation}
(\mathbf{I}-\mathbf{Q})\bar{\mathbf T}=(1,1,\ldots,1)^t.
\label{eqn:matrix}
\end{equation}
where $\mathbf Q$ is a $(n-1)\times (n-1)$ matrix, whose value in the $a$th row and $a'$th column is $Pr[B(l)=2a-a']$, $a,a' \in \{-(n-m)+1,\ldots,m-1\}$ ($\mathbf I$ is a commensurate identity matrix).
\end{theorem}

\begin{proof}
We show that at each iteration in the recursive description of the protocol, the problem is to compute $f(\tilde{\bX})$, where $\tilde{\bX}$ has the same joint probability distribution as $\bX$, perhaps with a different threshold. Towards this end it suffices to show (i) that the probability distribution $P_j$ of $X(i,j),~i=1,2,\ldots,n$ does not depend on the iteration index $j$, and that (ii) the possible value that a  threshold can have at any iteration in the algorithm lies in a finite set. Eq.~\eqref{eqn:matrix} then follows by a probability transition calculation. To see that (i) holds, note that $P_j$ is iid with marginal distribution that is uniform on $[0,1)$. For the proof of (ii), note that by \eqref{eqn:threshold}, if $a$ is an integer, then so is $a(j),j\geq 1$. The remaining steps to show (ii)  are by induction. For $j=1$, $a(j)$ lies in the set $\{-(n-m)+1,-(n-m)+2,\ldots,m-1\}$. Assume it is true for $j$. From the recursion \eqref{eqn:threshold}, and the fact that the protocol did not stop at iteration $j$, $B(j)$ must satisfy $2a_{j} -m < B(j) < 2a_{j}+(n-m)$, from which the assertion follows directly. Thus if the protocol continues at iteration $j$, the probability that the threshold  $a(j+1)=a'$ given that $a(j)=a$ is $Q(a,a')=Pr[B(j)=2a-a']$, which does not depend on $j$. Thus 
\begin{equation}
\bar{T}(a)=1+\sum_{a'=-(n-m)+1}^{m-1}\bar{T}(a') Q(a,a'),~a\in\{-(n-m),\ldots,m\},
\end{equation}
and \eqref{eqn:matrix} follows.
\end{proof}
\begin{example}
For $n=4$, $m=2$
\begin{equation}
{\mathbf Q}=\frac{1}{16}\left(\begin{array}{ccc}
4 & 1 & 0 \\
4 & 6 & 4 \\
0 & 1 & 4
\end{array}\right).
\end{equation}
\end{example}



\subsection{An approximation for $\bar{T}$}
The event $\{T>k\}$ is equivalent to the event $\{L(l)<Z(l) <U(l),~l=1,2,\ldots,k\}$, where $L(l),U(l)$ are defined immediately after \eqref{eqn:continue}. Thus $Pr[T>k] \leq Pr[L(k) < Z(k) < U(k)]$.
We will assume that $m/n =\beta$, a constant and that $a/n=\alpha$. By the central limit theorem, the cumulative distribution function of $Z(j)/\sqrt{n}$ converges  to that of a Gaussian with mean $\sqrt{n}(\beta-1/2)(1-2^{-j})$ and variance $(1-4^{-j})/12$.
 With $\sigma=\sqrt{n(1-4^{-k})/12}$, $\gamma=(\alpha+(1/2-\beta))$  let
 $A(n,k,\gamma)=\frac{1}{2}\left[\erf\left(\frac{(\gamma+ 2^{-(k+1)})n}{\sqrt{2}\sigma}\right) -   
\erf\left(\frac{(\gamma- 2^{-(k+1)})n}{\sqrt{2}\sigma}\right)\right].$
Then by the central limit theorem
$P[L(k)<Z(k)<U(k)]=A(n,k,\gamma)+\epsilon(n,k,\gamma)$, where $\epsilon(n,k,\gamma)$ is an error term in the central limit approximation. Thus
\begin{eqnarray}
\lim_{n\to \infty}\bar{T} & \leq  &1+\lim_{n\to \infty} \sum_{k=1}^\infty P[L(k)<Z(k)<U(k)] \nonumber \\
& = & 1+\lim_{n\to\infty}\sum_{k=1}^\infty \left(A(n,k,\gamma)+\epsilon(n,k,\gamma)\right). \nonumber \\
\label{eqn:clappxsum}
\end{eqnarray}
The first term $\sum_{k=1}^\infty A(n,k,\gamma) \leq A(n,\gamma)$, where
\begin{eqnarray}
\lefteqn{A(n,\gamma)= }\nonumber \\
&\left\{
\begin{array}{cc}
\frac{1}{2}\log_2\left( \frac{6n}{\pi}+1\right)+\sqrt{\frac{3}{1+\pi/(6n)}}, & \gamma = 0, \\
1+\kappa + o(1),  & 0 < \gamma \leq 1/4, \\
1+ o(1), & 1/4 < \gamma \leq 1/2,
\end{array}
\right.
\label{eqn:GaussAppx}
\end{eqnarray}  
where $o(1)\to 0$ as $n\to \infty$, and $\kappa$ is the largest integer for which $\frac{(\gamma-2^{-(\kappa+1)})n}{\sqrt{2}\sigma} \leq  -\frac{\sqrt{\pi}}{2}$. Note that $\kappa$ is bounded in $n$.
Eq.~\eqref{eqn:GaussAppx} follows  by applying  the bound $\erf(x) \leq \min(2x/\sqrt{\pi},1),~x\geq 0$ for the upper branch, $\erf(x)-\erf(y) \leq (2(x-y)/\sqrt{\pi})e^{-y^2},~x>y>0$, for the lower two branches. For the middle branch, we additionally use the bound $\erf(x)-\erf(y) \leq 2$, when $x>\sqrt{\pi}/2$ and $y<-\sqrt{\pi}/2$.
\begin{theorem}
$ \lim_{n\to \infty} \sum_{k=1}^\infty \epsilon(n,k,\gamma)=0$.
\end{theorem}
\begin{proof}
Since the random variables $b(i,l)$ have finite absolute third  moment, it follows from a result due to Berry and Esseen~\cite{butzer1975rate}, that for $x\in \mathbb{R}$, 
$|F_{Z(k)/\sqrt{n}}(x)-F_G(x)|\leq C_1/\sqrt{n}$, where 
$F_G$ is the CDF of a Gaussian random variable with with mean $\sqrt{n}(\beta-1/2)(1-2^{-k})$ and variance $(1-4^{-k})/12$, $F_{Z(k)/\sqrt{n}}$ is the CDF of $Z(k)/\sqrt{n}$ and $C_1$ is a constant independent of $x$. Further, $|F_G(U(k)/\sqrt{n})-F_G(L(k)/\sqrt{n})| \leq C_2\sqrt{n}2^{-k}$ and  $|F_{Z(k)/\sqrt{n}}(U(k)/\sqrt{n})-F_{Z(k)/\sqrt{n}}(L(k)/\sqrt{n})| \leq C_3\sqrt{n}2^{-k}$ for constants $C_2$ and $C_3$.  Thus $\epsilon(n,k,\gamma) \leq \min (C \sqrt{n}2^{-k},D/\sqrt{n})$ for constants $C$ and $D$. We now proceed to bound  $\sum_{k=1}^\infty \epsilon(n,k,\gamma)$. Let $k^*$ be the largest such that $c\sqrt{N} 2^{-k} > D/\sqrt{n}$ for $k < k^*$. Then   $\sum_{k=1}^\infty \epsilon(n,k,\gamma) \leq k^* D/\sqrt{n} +2D/\sqrt{n} \leq (1+o(1))\log n/\sqrt{n} $.  
\end{proof}
\begin{remark}
It is interesting that  when $\gamma \neq 0$, $\bar{T}$ is integer valued.
\end{remark}

\section{One-Way Protocols}
\label{sec:SingleRound}
In the One-Way protocol, the class label is to be known at the leader node alone. In the One-Way$+$ protocol, the leader  informs each sensor node of the class label at an additional cost of $n$ bits.  The encoder for the $i$th node is a mapping $\phi~:~\mathbb{R}\rightarrow \mathcal{C}_i$, where $\mathcal{C}_i$ is the codebook for node $i$. Encoded value $\phi(X_i)$ is sent to a leader node using  $R_i=\log_2(|\mathcal{C}_i|)$ bits. The leader node estimates the class label, and informs each node about this. The total rate is given by $R_{sum}=\sum_{i=1}^n R_i +n$ bits. Suppose the estimated class label is $g(\bX)$. The error $P_e=Pr[f(\bX)\neq g(\bX)]$ is a function of the rate $R_{sum}$. We analyze and implement this protocol when $\phi(X_i)=\lfloor X_i/\delta \rfloor \delta$ for each $i$ ($\lfloor x\rfloor$ is the largest integer $\leq x$). In terms of the step size $\delta$, $R_{sum}=n(\log_2(1/\delta)+1)$ bits. We now determine $\delta$, and thus $R_{sum}$  in terms of $P_e$. Our main finding is that for $m=n-m$ and $a=0$, the sum rate for One-Way$+$ is
\begin{equation}
R_{sum} \geq n\left(\log_2\left(\frac{1}{P^*_e}\right)+\frac{1}{2}\log_2 \left( \frac{12n}{\pi}\right)+1\right),
\label{eqn:OW}
\end{equation}
where $P^*_e$ is a target upper bound on the error probability.
Thus for this case, for fixed $P_e$, $R_{sum}/n$ grows logarithmically in $n$.

The proof of  \eqref{eqn:OW} is now presented. Observe that the protocol partitions the unit cube $[0,1)^n$ into smaller cubes, called cells,  of side $\delta$. An error occurs only when the source vector lies in a cell whose interior does not intersect the separating hyperplane $\partial f :=\{\sum_{i\in \cH^+}X_i-\sum_{i\in\cH^-}X_i=a\}$. We now calculate the number of cells that intersect $\partial f$. In order to simplify the analysis a bit, let $Y_i=-X_i+1$ for $i\in \cH^-$. Note that, like $X_i$, $Y_i$ is uniformly distributed on $[0,1)$. Now rename the $Y_i$ to $X_i$. The boundary $\partial f :=\{\sum_i X_i=a+(n-m)\}$. Thus under the transformation, the threshold is now $a+(n-m)$, which we rename to $a$, with the constraint that the new $a$ satisfies $0< a < n$. Under our assumption of uniform quantization, the lower endpoints of the bins are in the set ${\mathcal L}:=\{m\delta$, $m=0,1,\ldots,M-1\}$, where $M=1/\delta$ is the number of bins. Let $L_i=l_i\delta\in\mathcal L$ be the lower endpoint of the bin for $X_i$, i.e $l_i \leq X_i < l_i+\delta$. Then a cell intersects $\partial f$ if and only if $a- n\delta \leq \sum_i L_i < a$, or equivalently
\begin{equation}
-\frac{1}{2} < \frac{1}{n}\sum_i l_i -\left(\frac{M-1}{2}\right) < \frac{1}{2},
\end{equation}
as can be seen after some algebraic manipulation. The random variable $Z:=\frac{1}{n}\sum_i l_i -\left(\frac{M-1}{2}\right)$ has mean zero and variance $(M^2-1)/(12n)$ and through an application of the central limit theorem it follows that $\lim{n\to\infty}P_e=\lim_{n\to\infty}(1/2)\erf(\sqrt{12n/(M^2-1)})$. The bound $\erf(x) \leq \min(2x/\sqrt{\pi},1),~x\geq 0$, and some further algebra leads to \eqref{eqn:OW}.

We conjecture that the performance of a uniform quantizer cannot be improved on, under the assumptions made here. The proof of this appears to be significantly complicated. 

\section{Lower Bounds for the Interactive Protocol}
\label{sec:converse}

We first show, through a lower bounding argument, that for $n=2, m=1$ and  $a=0$, the bit exchange protocol achieves the optimal sum rate. For $n>2$, the protocol achieves optimal scaling with respect to $n$, when $\gamma\neq 0$. However, when $\gamma \neq 0$, there is a $\log n$ gap in the sum rate.  We show this through a lower bound for $n>2$.

\subsection{A Tight Lower Bound for  $n=2,m=1,a=0$}


The Interactive Protocol results in a rectangular partition $(\cP,\cQ)$ where  $\cP$ is a subpartition of the region $\cR_p=\{\bx~:~f(\bx)=1\}$ and $\cQ$ of $\cR_p=\{\bx~:~f(\bx)=0\}$. Since the error probability is zero, we refer to $(\cP,\cQ)$ as a \emph{zero-error} partition.  The boundary is represented by the set $\cB:=\{x_1=x_2\}\bigcap (0,1)^2$.  For a zero-error partition $(\cP,\cQ)$ it is true that each point of $\cB$  must be the upper left corner of some rectangle that lies entirely in $\cR_p$ or the lower right corner of some rectangle that lies entirely in $\cR_q$, except possibly for a set of one-dimensional measure zero. Let $\{p_i,~i=1,2,\ldots\}$ be the probabilities of  the   cells of $\cP$   and let $\{q_i,~i=1,2,\ldots\}$ be the  probabilities of the  cells of $\cQ$. We note here that we use the word rectangle to include what is referred to as a \emph{combinatorial rectangle}, which is the Cartesian product of unions of intervals~\cite{KushNis:1997}.

\begin{figure}[htbp] 
   \centering
   \begin{tabular}{cc}
   \includegraphics[width=2.6cm]{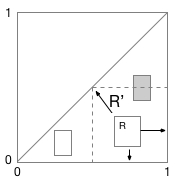} &
   		\includegraphics[height=2.6cm]{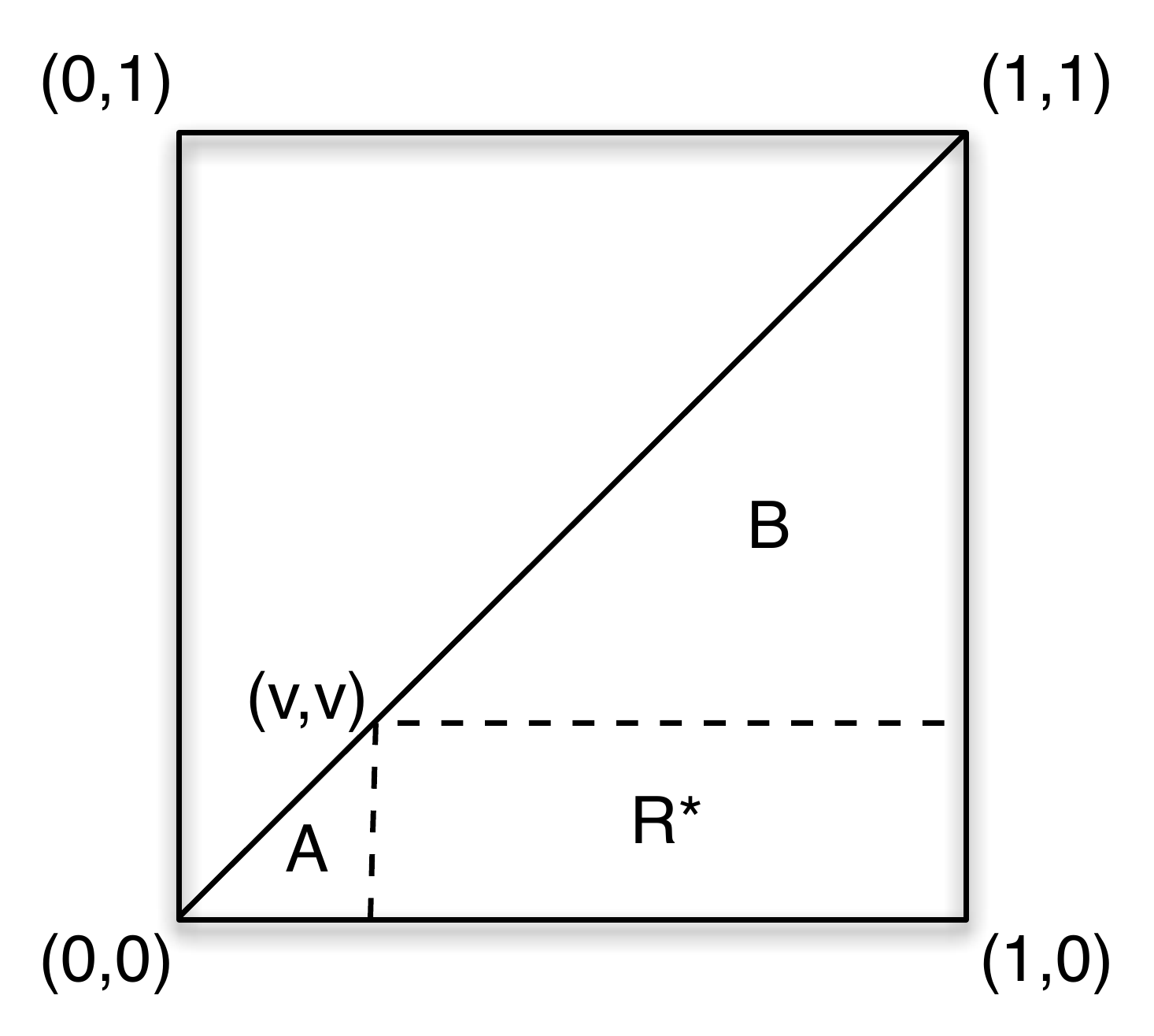}  
		\end{tabular} 

   \caption{(left) Illustration of the readjustment of the rectangle $R$ with the highest probability, (right) Regions used in Thm.~\ref{thm:fourbits}. $R_p=A\bigcup R^* \bigcup B$.}
   \label{fig:Readjust}
\end{figure}

\begin{theorem}
\label{thm:pathconnected}
If a partition minimizes the entropy it contains a rectangle with vertices $(1,0)$ and   $(v,v)$ and  another rectangle with vertices $(0,1)$ and  $(u,u)$, for some $0<u,v <1$. 
\end{theorem}

\begin{proof} (Sketch)
The idea is to grow the largest rectangle in a partition cell until a vertex touches the boundary while reducing the probabilities of the other rectangles in the partition cell. The probability distribution of the readjusted rectangles majorizes all other rectangular partitions. The details of the proof are omitted due to space constraints and can be found in~\cite{VV:2017}.
\end{proof}
\begin{theorem}
For $n=2,m=1,a=0$, the sum rate is no smaller than $4$ bits.
\label{thm:fourbits}
\end{theorem}
\begin{proof}
Consider an extreme partition $(\cP,\cQ)$ which contains a rectangle $R^*$ which has a points $(v,v)$ and $(1,0)$ as its upper left and lower right vertices.  This is always true by Thm.~\ref{thm:pathconnected}. Let random variable $C$ indicate whether $(x_1,x_2)$ lies in $\cR_p$ or not, and let random variable $S$ indicate whether $(x_1,x_2)$ lies in one of the three regions, $R^*$, $A$ or $B$ as shown in Fig.~\ref{fig:Readjust}. Let $H(\cP,\cQ))$ denote the entropy of the partition $(\cP,\cQ)$. Then 
\begin{eqnarray}
H(\cP,\cQ) & = & H(C)+H(\cP,\cQ|C=0)P(C=0)+ \nonumber \\
& & ~~H(\cP,\cQ|C=1)P(C=1) 
\label{eqn:firstdecomp}
\end{eqnarray}
and
\begin{eqnarray}
\lefteqn{H(\cP,\cQ|C=1)= H(S|C=1)+} & \nonumber \\ 
& +H(\cP,\cQ|C=1,S=A)P(S=A|C=1)+ \nonumber \\
& + H(\cP,\cQ|C=1,S=B)P(S=B|C=1). 
\end{eqnarray}
Since the regions $A$ and $B$ are similar to $\cR_p$ it follows that if this partition minimizes the entropy it must satisfy the recursion
\begin{eqnarray}
\lefteqn{H(\cP,\cQ|C=1)  =  H([v^2,2v(1-v),(1-v)^2])+} & & \nonumber \\
&  & ~~~~~+ (v^2+(1-v)^2)H(\cP,\cQ|C=1).
\end{eqnarray}
Solving for $H(\cP,\cQ|C=1)$ we obtain
\begin{equation}
H(\cP,\cQ|C=1)=\frac{H([v^2,2v(1-v),(1-v)^2])}{2v(1-v)}
\label{eqn:entropyratio}
\end{equation}
whose unique minimum value of $3$ bits occurs when $u=v=1/2$. Plugging back in (\ref{eqn:firstdecomp}) leads to the desired result.
\end{proof}


\subsection{$n>2$}
\label{sec:multiparty}

\begin{lemma} (Lemma 9 in \cite{OrEl:1990})
Assume that probabilities $p_i,~i=1,2,\ldots$ are in decreasing order, $\sum_{i=1}^\infty p_i =1$ and that $\sum_{i=m}^\infty p_i \geq g(m)$, $m=1,2,\ldots$. Then the entropy $H(p) \geq \sum_{i=1}^{m-1}p_i\log (1/p_i) + g(m) \log (m/(1-g(m+1)))$.
\end{lemma}

It turns out to be  easy to determine the largest rectangle in $f^{-1}(0)$ for $n>2$. We do so here for the case 
 where $a=0$ and $m=n/2$. In this case, with $p_1$ being the probability of the largest rectangle in $f^{-1}(0)$ it follows that  $p_1 \leq  2^{-n}$ as can be checked by maximizing $\prod_{i\in \cH^+ }(1-x_i)\prod_{i \in \cH^-} x_i$ subject to $\sum_{i\in \cH^+}x_i -\sum_{i\in \cH^-}x_i =0$. Thus it follows that $H(p) \geq n$. In this specific case, there is a $\log n$ gap between the upper bound of our protocol and the lower bound. For the cases where $\gamma >0$, the rate grows linearly in $n$. Since we cannot expect a slower growth with $n$ (the result must be distributed to $n$ nodes), the protocol exhibits the correct scaling behavior with $n$.  
\section{Numerical Experiments}
\label{sec:numeric}
\begin{figure}[htbp] 
	\centering
	\begin{tabular}{cc}
	\includegraphics[width=1.7in]{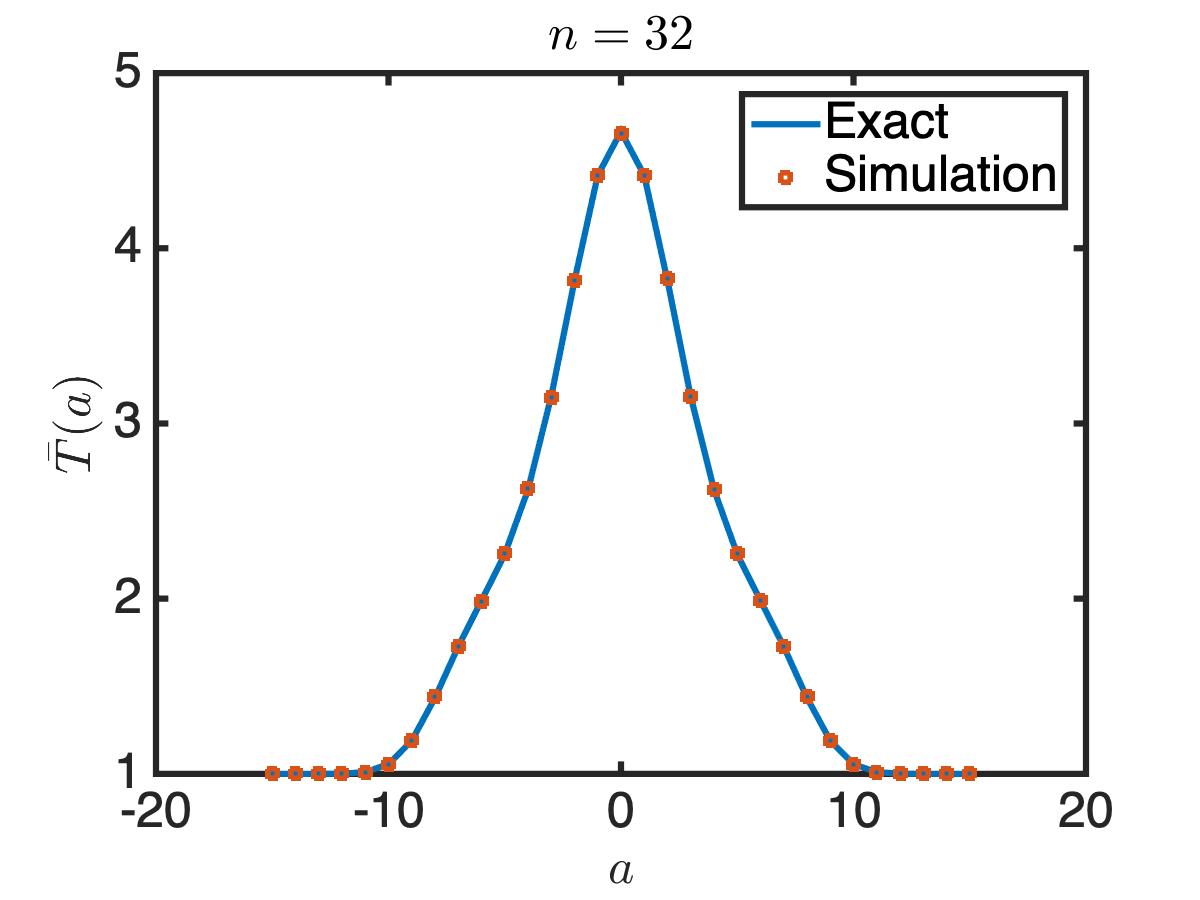} 	& \includegraphics[width=1.7in]{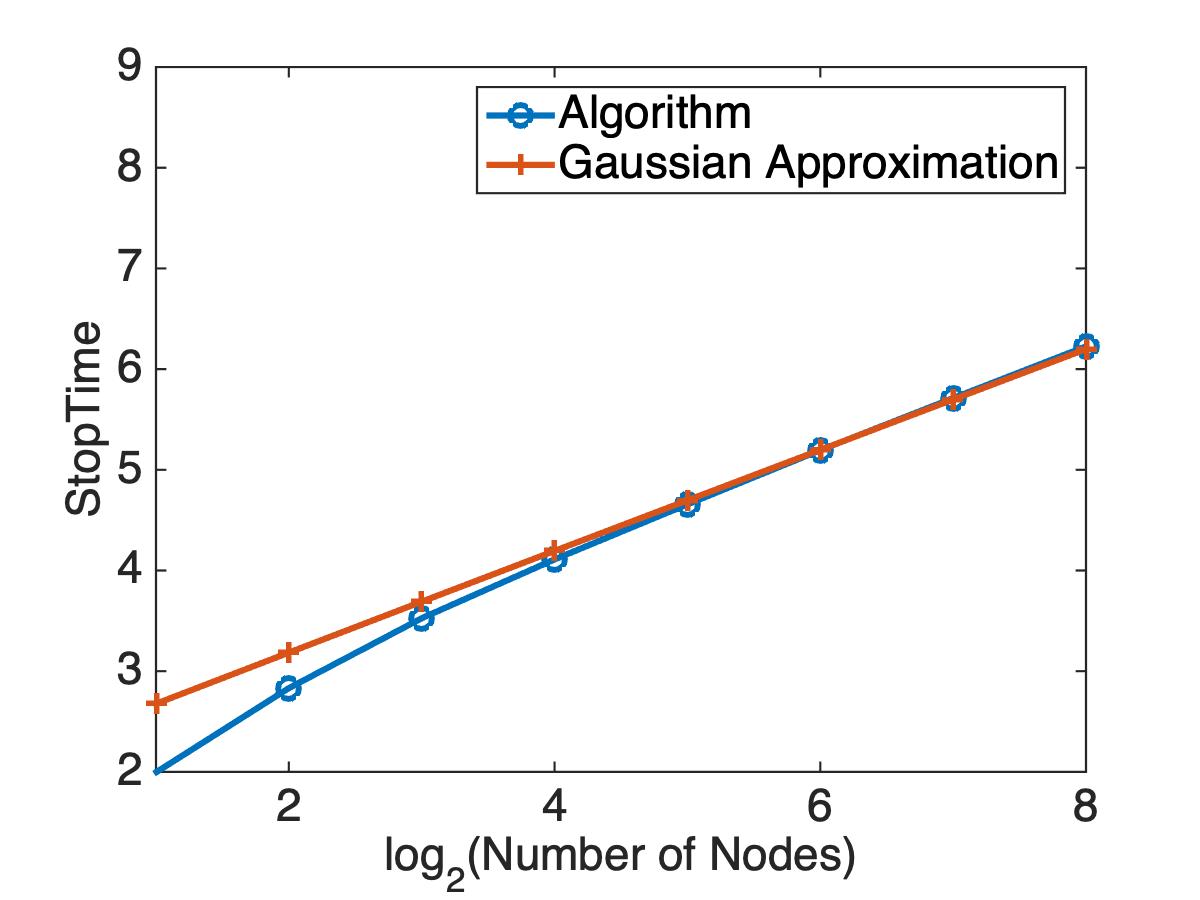}\\  \includegraphics[width=1.7in]{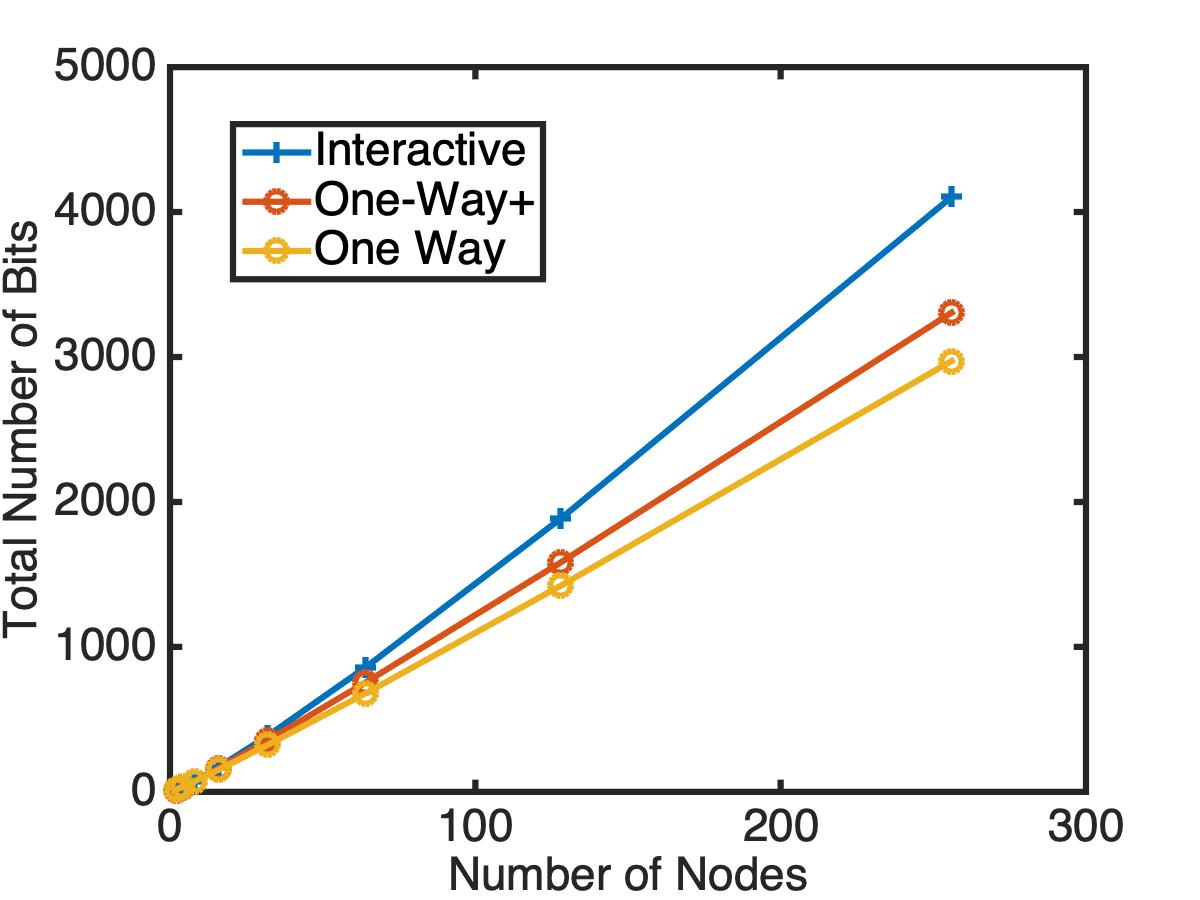}&   \includegraphics[width=1.7in]{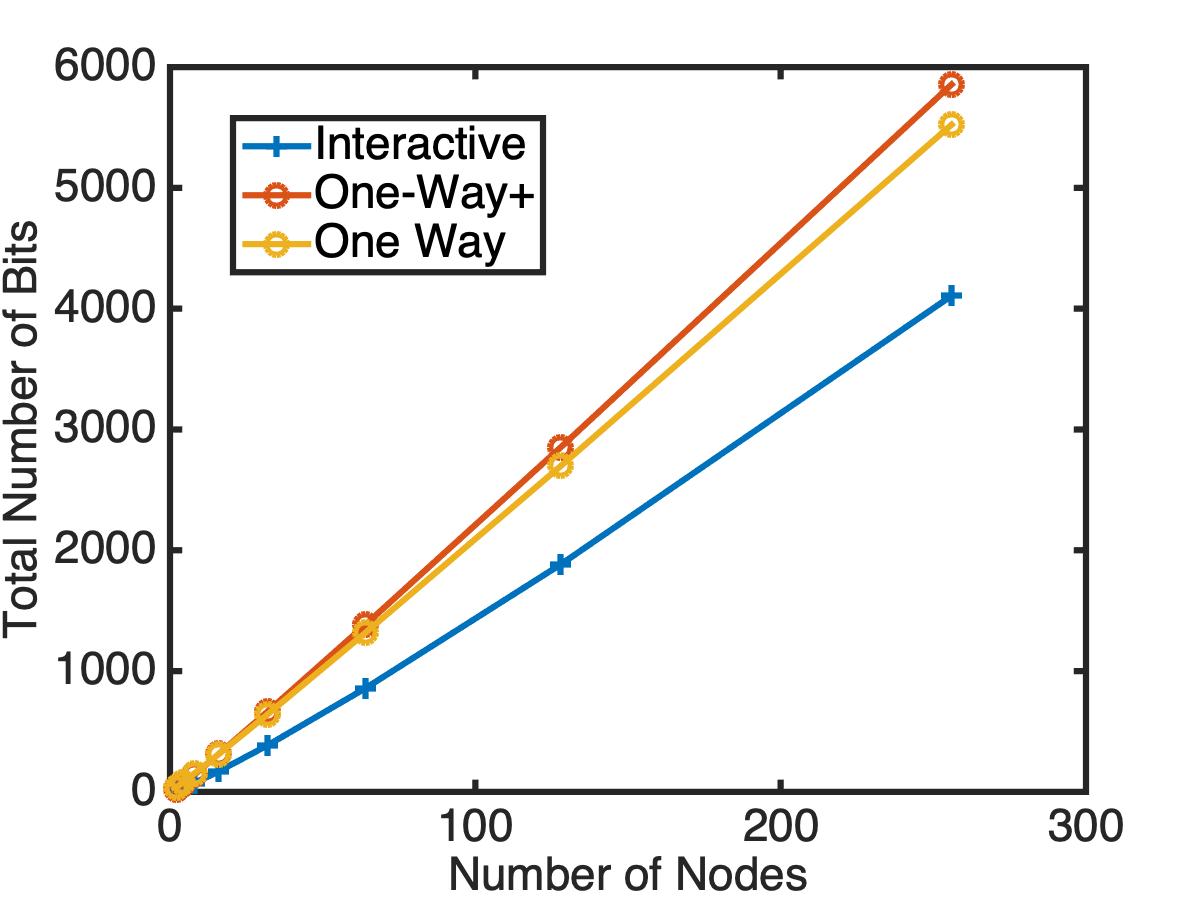}
	\end{tabular}
	\caption{From top left to bottom right: (a) Comparison of exact \eqref{eqn:matrix} and simulation results for $n=32$  (b) stopping time of the Interactive Protocol as a function of $n$: simulation of actual protocol and Gaussian approximation \eqref{eqn:GaussAppx}, $a=0$ and $m=n/2$,(c)-(d) sum rate of the One-Way, One-Way$+$ and Interactive protocol for  $P_e=10^{-2}$ and  $P_e=10^{-5}$, resp. }
	\label{fig:StopTime}
\end{figure}

Computer simulations were carried out for the Interactive Protocol for a $n=32$ node network and compared to \eqref{eqn:matrix}. Simulations were based on $20,000$ repetitions per data point in the graph shown in Fig.~\ref{fig:StopTime}(a).  Fig.~\ref{fig:StopTime}(b) shows the agreement between the growth of the stopping time obtained via simulation and with the approximation \eqref{eqn:GaussAppx} for $m=n/2$, $a=0$. The agreement is good in all cases. 

Fig.~\ref{fig:StopTime}(d) compares the sum  rate of the Interactive Protocol with the One-Way protocols for $m/n=1/2$ and $a=0$.
In our experiments, we used a uniform quantizer for quantizing $X_i$ for each $i$. Note that in the one-way case, it is not possible to obtain zero error with finite rate. In our experiments we set the error probability to $10^{-5}$.
It is seen that for a $n=256$ node network, the Interactive Protocol used $4106$ bits whereas the One-Way protocol uses $6374$ bits; a non-trivial saving of almost $35$\% over the One-Way protocol for $P_e=10^{-5}$. On the other hand, as seen in Fig.~\ref{fig:StopTime}(c), when $P_e=10^{-2}$, the One-Way protocol uses only 2972 bits and thus the Interactive Protocol uses 38\% more bits. In both cases, the Interactive Protocol achieves this with zero error, at the cost of variable completion time. 
\section{Summary and Conclusions}
\label{sec:summary}
We presented and analyzed interactive and one-way protocols for solving a 2-classifier problem in a sensor network with $n$ sensor nodes. The sum rate of the interactive protocol depends on the mean stopping time of the protocol. Analysis  of both protocols is presented. The analysis reveals a $\log n$ gap between the interactive protocol and the lower bound. The relative performance of the interactive and one-way protocols is seen to depend strongly on the error probability required of the one-way protocol.
The cause of the gap between the interactive protocol and its lower bound requires further investigation. 

\end{document}